\documentclass[a4paper,UKenglish,cleveref, autoref, thm-restate]{lipics-v2021} % test
\usepackage{todonotes}
\usepackage[fixamsmath,disallowspaces]{mathtools}
\usepackage{fixmath}
\usepackage{bm}
\usepackage{amsmath}
\usepackage{todonotes}
\usepackage{xspace}
\usepackage{booktabs}
\newcommand{\complClFont}[1]{\normalfont\textbf{#1}}         % P, NP, NL, ...
\newcommand{\logicClFont}[1]{\mathcal{#1}}        % PDL, MSO, ...
\newcommand{\problemFont}[1]{\mathrm{#1}}         % SAT, TAUT, IMP, ...
\usepackage{microtype}
\usepackage{tikz}
\usetikzlibrary{decorations.pathreplacing}
\usetikzlibrary{arrows}

\usepackage{mathtools}
\DeclarePairedDelimiterX\set[1]\lbrace\rbrace{#1}

\newcommand{\Vars}{\protect\ensuremath{\mathrm{Fr}}}
\newcommand{\SubForm}[1]{\mathrm{SF}\ifx#1\empty\else(#1)\fi}

\newcommand{\var}[1]{\mathsf{VAR}\ifx#1\empty\else(#1)\fi}
\newcommand{\SF}[1]{\mathsf{SF}\ifx#1\empty\else(#1)\fi}
\newcommand{\multiple}[1]{\protect\ensuremath{f_\#\ifx#1\empty\else(#1)\fi}}

\newcommand{\depa}[2]{{\mathsf{dep}}({\mathbf{#1}};{{#2}})}

\newcommand{\depas}[2]{{\mathsf{dep}}({#1};{#2})}
\newcommand{\inca}[2]{{\mathbf{#1}}\subseteq\mathbf{#2}}

\newcommand{\indepa}[3]{{\mathbf{#1}}\bot_{\mathbf{#3}}\mathbf{#2}}

\newcommand{\para}{\protect\ensuremath{\complClFont{para}}} 
\newcommand{\FPT}{\protect\ensuremath{\complClFont{FPT}}} 
\newcommand{\NP}{\protect\ensuremath{\complClFont{NP}}} 
 
\newcommand{\Ptime}{\protect\ensuremath{\complClFont{P}}} 
\newcommand{\XP}{\protect\ensuremath{\complClFont{XP}}}
\newcommand{\WP}{\protect\ensuremath{\complClFont{W[P]}}}

\newcommand{\W}[1]{\protect\ensuremath{\complClFont{W}\ifx#1\empty\else[#1]\fi}}

\newcommand{\p}{\text{p-}}
\newcommand{\rel}{\mathrm{rel}}

\newcommand{\calA}{\mathcal{A}}
\newcommand{\calG}{\mathcal{G}}
\newcommand{\calC}{\mathcal{C}}

\newcommand{\preduction}{\ensuremath{\leq^{\Ptime}_m}}
\newcommand{\fptreduction}{\ensuremath{\leq^{\FPT}}}

\newcommand{\dfn}{\mathrel{\mathop:}=}

\newcommand{\FO}{\protect\ensuremath{\logicClFont{FO}}\xspace}
\newcommand{\Fragment}{\protect\ensuremath{\mathrm{\Theta}}}

\newcommand{\PiForm}[1]{\protect\ensuremath{\mathrm{\Pi}_{#1}}}
\newcommand{\SigForm}[1]{\protect\ensuremath{\mathrm{\Sigma}_{#1}}}

\newcommand{\IND}{\protect\ensuremath{\logicClFont{FO}(\bot)}\xspace}
\newcommand{\INC}{\protect\ensuremath{\logicClFont{FO}(\subseteq)}\xspace}

\newcommand{\D}{\protect\ensuremath{\logicClFont{FO}(\mathsf{dep})}\xspace}

\newcommand{\DPi}[1]{\protect\ensuremath{\D\text-\PiForm{#1}}\xspace}
\newcommand{\INCPi}[1]{\protect\ensuremath{\INC\text-\PiForm{#1}}\xspace}
\newcommand{\INDPi}[1]{\protect\ensuremath{\IND\text-\PiForm{#1}}\xspace}

\newcommand{\DSigma}[1]{\protect\ensuremath{\D\text-\SigForm{#1}}\xspace}
\newcommand{\INCSigma}[1]{\protect\ensuremath{\INC\text-\SigForm{#1}}\xspace}
\newcommand{\INDSigma}[1]{\protect\ensuremath{\IND\text-\SigForm{#1}}\xspace}

\newcommand{\ESO}{\protect\ensuremath{\logicClFont{ESO}}\xspace}

\newcommand{\WD}{\problemFont{WD}}
\newcommand{\WTD}{\problemFont{WT}}
\newcommand{\FD}{\problemFont{FD}}

\newcommand{\IndSet}{\protect\ensuremath{\textsc{IndependentSet}}\xspace}
\newcommand{\DomSet}{\protect\ensuremath{\textsc{DominatingSet}}\xspace}
\newcommand{\Clique}{\protect\ensuremath{\textsc{Clique}}\xspace}
\newcommand{\wsatpos}[2]{\protect\ensuremath{\problemFont{WSAT}(\Gamma^+_{#1}\ifx#2\empty\else,#2\fi)}} 
\newcommand{\wsatneg}[1]{\protect\ensuremath{\problemFont{WSAT}(\Gamma^-_{#1})}}
\newcommand{\wsat}[1]{\protect\ensuremath{\problemFont{WSAT}(\Gamma_{#1})}} 
\newcommand{\WSAT}{\protect\ensuremath{\problemFont{WSAT}}}

\newcommand{\problemdef}[3]{%
\begin{center}
\begin{tabular}{lp{10cm}}\toprule
\textsf{\bfseries Problem:}& #1 \\\midrule
\textsf{\bfseries Input:}& #2.\\
\textsf{\bfseries Question:}& #3?\\\bottomrule
\end{tabular}
\end{center}
}

\newcommand{\paraproblemdef}[4]{%
\begin{center}
\begin{tabular}{lp{10cm}}\toprule
\normalfont\textsf{\bfseries Problem:}& \normalfont#1 \\\midrule
\normalfont\textsf{\bfseries Input:}& \normalfont#2.\\
\normalfont\textsf{\bfseries Parameter:}& \normalfont#3.\\
\normalfont\textsf{\bfseries Question:}& \normalfont#4?\\\bottomrule
\end{tabular}
\end{center}
}

\newcommand{\VAR}{\mathrm{VAR}}
\newcommand{\Fr}{\Vars}
\newcommand{\dom}{\mathrm{dom}}

\newcommand{\tuple}[1]{\mathbf{#1}}

\hideLIPIcs  %uncomment to remove references to LIPIcs series (logo, DOI, ...), e.g. when preparing a pre-final version to be uploaded to arXiv or another public repository

\title{Parameterized Complexity of Weighted Team Definability}
\author{Juha Kontinen}{Department of Mathematics and Statistics, University of Helsinki, Finland}{juha.kontinen@helsinki.fi}{https://orcid.org/0000-0003-0115-5154}{Partially funded by Academy of Finland grant 338259}
\author{Yasir Mahmood}{Institut f\"{u}r Theoretische Informatik, Leibniz Universit\"{a}t Hannover, Germany}{mahmood@thi.uni-hannover.de}{https://orcid.org/0000-0002-5651-5391}{}
\author{Arne Meier}{Institut f\"{u}r Theoretische Informatik, Leibniz Universit\"{a}t Hannover, Germany}{meier@thi.uni-hannover.de}{https://orcid.org/0000-0002-8061-5376}{Partially funded by DFG grant ME 4279/3-1}
\author{Heribert Vollmer}{Institut f\"{u}r Theoretische Informatik, Leibniz Universitf\"{a}t Hannover, Germany}{vollmer@thi.uni-hannover.de}{https://orcid.org/0000-0002-9292-1960}{Partially funded by DAAD Project-ID 57570031}

\authorrunning{J. Kontinen Y. Mahmood, A. Meier, and H. Vollmer } %TODO mandatory. First: Use abbreviated first/middle names. Second (only in severe cases): Use first author plus 'et al.'

\Copyright{Juha Kontinen, Yasir Mahmood, Arne Meier, and Heribert Vollmer} %TODO mandatory, please use full first names. LIPIcs license is "CC-BY";  http://creativecommons.org/licenses/by/3.0/


 %TODO mandatory; please add comma-separated list of keywords

\nolinenumbers %uncomment to disable line numbering

\usepackage{mathtools}

\begin{document}

\maketitle

\begin{abstract}
In this article, we study the complexity of weighted team definability for logics with team semantics. 
This problem is a natural analogue of one of the most studied problems in parameterized complexity, the notion of weighted Fagin-definability, which is formulated in terms of satisfaction of first-order formulas with free relation variables.
We focus on the parameterized complexity of weighted team definability for a fixed formula $\varphi$ of central team-based logics.
Given a first-order structure $\cal A$ and the parameter value $k\in \mathbb N$ as input, the question is to determine  whether $\calA,T\models \varphi$ for some team $T$ of size $k$. 
We show several results on  the complexity of this problem for dependence, independence, and inclusion logic formulas.
Moreover, we also relate the complexity of weighted team definability to the complexity classes in the well-known W-hierarchy as well as paraNP. 
\end{abstract}

\section{Introduction}
In this article, we study the parameterized complexity of weighted team definability for logics in team semantics. 
Team definability is a natural analogue of the notion  of Fagin-definability whose weighted version can be used to characterize the W-hierarchy in parameterized complexity~\cite{DFR98}. 
We give several results on the complexity of this problem for dependence, independence, and inclusion logic formulas.

The birth of the nowadays established logics of dependence and independence can be traced back to the  introduction of dependence logic in 2007~\cite{DBLP:books/daglib/0030191}. 
In team semantics, formulas are interpreted by sets of assignments (teams) instead of a single assignment as in Tarski's semantics of first-order logic. 
Syntactically dependence logic extends first-order logic by new dependence atomic formulas (dependence atoms) $\depa{x}{y}$ expressing that the values of variables $\mathbf x$ functionally determine the value of the variable $y$ in the team under consideration. 
Independence and inclusion logics are further extensions of first-order logic by independence atoms $\indepa{x}{y}{z}$ and inclusion atoms $\mathbf x \subseteq \mathbf y$ which essentially correspond to embedded multivalued dependences and inclusion dependences from database theory~\cite{gradel10,galliani12}.  

For the applications, it is important to understand the complexity theoretic aspects of team-based logics.  
During the past ten years, the expressivity and complexity theoretic aspects of logics in first-order (also propositional~\cite{DBLP:journals/apal/YangV17}, modal~\cite{DBLP:journals/logcom/HellaKMV19,DBLP:journals/tocl/HellaKMV20}, temporal~\cite{DBLP:conf/lics/GutsfeldMOV22} and probabilistic~\cite{DBLP:conf/foiks/0001HKMV18}) team semantics have been studies extensively (see, e.g., \cite{HannulaKVV18,Luck19,HannulaKBV20,DurandKRV22}). 
The baseline for these studies are the well-known results stating that the sentences of dependence logic and independence logic are equivalent to existential second-order logic while inclusion logic corresponds to positive greatest fixed point logic and thereby captures $\Ptime$ over finite (ordered) structures~\cite{gallhella13}. 
In team semantics results for sentences of the logic do not immediately extend to open formulas. 
In particular, the open formulas of dependence logic correspond in expressive power to sentences of $\ESO$ with an extra relation encoding the team that occurs  only negatively in the sentence~\cite{KontinenV09}. 
For independence logic, the requirement of negativity can be lifted~\cite{galliani12}. For inclusion logic an analogous result shows that any first-order sentence $\varphi(R)$ whose truth is preserved under $R$-unions  can be expressed by an inclusion logic formula $\varphi^*(\mathbf x)$. 
In other words, for all $\calA$ and teams $T\neq \emptyset$:
\[\calA, T \models \varphi^*(\mathbf x)  \Leftrightarrow \calA \models \varphi(\rel(T)/R),\]
where $\rel(T)$ is a relation encoding the team $T$ \cite{gallhella13}. 
These result can be  used to relate weighted team definability to  weighted Fagin-definablity. 
However, it is instructive to note that, due to higher expressive power of the logics considered in this article, the syntactic complexity of a formula does not in general correlate with the complexity of the model-checking of the formula. 
In particular, any formula of dependence and independence logic is logically equivalent to a formula with $\forall\exists$-quantifier prefix~\cite[Theorem 6.15]{DBLP:books/daglib/0030191}~\cite[Theorem 4.9]{KontinenV09}.

A formalism to enhance the understanding of the inherent intractability of computational problems is brought by the framework of parameterized complexity~\cite{DBLP:series/txcs/DowneyF13}. 
%Founded by Downey and Fellows, in parameterized complexity theory one strives for more structure within a problem instance that render this problem tractable. 
Here, one aims to find parameters relevant for practice allowing to solve the problem by algorithms running in time $f(k)\cdot n^{O(1)}$, for some computable function $f$, where $k$ is the parameter value and $n$ is the input length.
Problems with such a running time are called \emph{fixed-parameter tractable} ($\FPT$) and correspond to efficient computation in the parameterized setting.
%A different quality of runtimes (of the form $n^{f(k)}$) are obeyed by algorithms solving problems in the class $\XP$.
%Comparing both classes with respect to the runtimes their problems allow to be solved in, of course, both runtimes are polynomial.
%However, for the first type, the degree of the polynomial is independent of the parameter's value which is notable to observe.
%As a result, the second kind of runtimes is undesirable and usually tried to circumvented by locating different parameters.
%
The problems solvable within the runtimes of the form $f(k)\cdot n^{O(1)}$ with respect to nondeterministic machines belong to the complexity class $\para\NP\supseteq\FPT$.
Moreover, restricting the amount of nondeterminism allows to study a subclass $\W\Ptime\subseteq \para\NP$. 
The complexity class $\W\Ptime$ is defined via nondeterministic machines that have at most $h(k)\cdot\log n$ many nondeterministic steps, where $h$ is a computable function. %, $k$ is the parameter value, and $n$ is the input length.
In between $\FPT$ and $\W\Ptime$, a presumably infinite $\complClFont{W}$-hierarchy is contained: $\FPT\subseteq\W1\subseteq\W2\subseteq\dots\subseteq\W\Ptime$.
It is unknown whether any of these inclusions is strict.
Showing $\W1$-hardness of a problem intuitively corresponds to being intractable in the parameterized world.

%Coming back to fpt-runtimes, a runtime of a very different quality (yet still polynomial for fixed parameters) than $\FPT$ is summarized by the complexity class $\XP$: $|x|^{f(k)}$ for inputs $x$, corresponding parameter values $k$, and a computable function $f$. 
%Furthermore, analogously as $\XP$ but on nondeterministic machines, the class $\XNP$ will be of interest in this paper.
%Further up in the hierarchy, classes of the form $\para\mathcal{C}$ for a classical complexity class $\mathcal{C}\in\{\NP,\PSPACE,\NEXP\}$ play a role in this paper.
%Such classes intuitively capture all problems that are in the complexity class $\mathcal C$ after fpt-time preprocessing. 
%In Fig.~\ref{fig:cc-landscape} an overview of these classes and their relations are depicted (for further details see, e.g., the work of Elberfeld~et~al.~\cite{DBLP:journals/algorithmica/ElberfeldST15}).

\begin{table}
	\centering
	\begin{tabular}{cllc}\toprule
		\bfseries Logic &\bfseries  $\exists\,\varphi_t$ s.t.\ $\p\WTD_{\varphi_t}$ is &\bfseries  Condition &\bfseries Result\\\midrule
		$\FO$ 
		&
		in TC$^0$ & unparameterized, all formulas& Thm.~\ref{thm:WT-FO}\\
		$\INC$
		& in $\FPT$/$\WP$ & all sentences / all formulas & Thm.~\ref{inc:sentences}/Thm.~\ref{INC:WP}\\ 
		& $\W1$-hard, $\in\W2$ & any quantifier-free formula, without $\lor$ & Cor.~\ref{cor:inc-lor}\\
		& $\W t$-complete
		& for all even $t\in\mathbb N$ & Cor.~\ref{cor:inc-wt}\\
		$\D$ & $\W t$-complete&for all odd $t\in\mathbb N$ & Cor.~\ref{cor:D-wt}\\
		& $\para\NP$-complete & sentence / formula & Thm.~\ref{D:sentences}/Thm.~\ref{thm:WT-D}\\
		$\IND$ 
		& $\W t$-complete & for all $t\in\mathbb N$ & Thm.~\ref{thm:WT-I} (1.)\\
		& $\W\Ptime$-complete& formula & Thm.~\ref{thm:WT-I} (2.)\\
		& $\para\NP$-complete& sentence / formula & Thm.~\ref{thm:WT-I} (3.) / (4.)\\
		\bottomrule
	\end{tabular}
	\caption{Partial overview of our results concerning weighted team definability with pointers to the respective theorem or corollary.}\label{tbl:overview}
\end{table}
\subparagraph*{Our contributions}
We define and study the parameterized complexity of weighted team definability with respect to formulas of several team-based logics.
Moreover, we establish the relationship between our framework and the problem of weighted Fagin definability.
%Our results show that for plain first-order formulas weighted team definability differs greatly from weighted Fagin definability the former being computationally much simpler. 
In more details, we explore the complexity of weighted team definability in parameterized setting for dependence, independence and inclusion logic formulas as well as sentences. 
Thereby, we prove and obtain novel logical characterizations of, and new complete problems for, the aforementioned central parameterized complexity classes, i.e., the $\W{}$-hierarchy, $\W\Ptime$, and $\para\NP$. % the complexity ranges between the classes $\W t$ and $\para\NP$. 
Table~\ref{tbl:overview} gives a partial overview of our results concerning weighted team definability. 
\subparagraph*{Related work} The complexity of counting/enumerating satisfying teams for a fixed first-order formula of team-based logic has been studied before \cite{HaakKMVY19,HaakMMV22}. 
Furthermore, there are also recent works on the parameterized complexity model-checking and satisfiability for propositional and first-order team-based logics \cite{DBLP:conf/foiks/MeierR18,DBLP:conf/foiks/0002M20, DBLP:journals/corr/abs-2105-14887,KontinenMM22}.
	Regarding the descriptive complexity, Downey~et~al.~\cite{DFR98} explored the logical characterization of the classes in the $\W{}$-hierarchy. 
%	Their results are analogous to the well-known Fagin's theorem for $\NP$.

\section{Preliminaries}
We require a basic knowledge of standard notions from classical complexity theory \cite{DBLP:books/daglib/0072413}. 
The classical complexity classes we encounter mostly in this work are $\Ptime$ and $ \NP$ together with their respective completeness notions, employing polynomial time many-one reductions ($\preduction $). 
Moreover, we assume the reader is familiar with the basic first-order (predicate) logic~\cite{DBLP:books/daglib/0082516}.
In the following, we define a few important classes of first-order formulas which are relevant to the results in this work.

\subparagraph*{FO-Formula Classes}
%We assume a vocabulary $\tau$ which includes relation and constant symbols.
The class of all first-order formulas is denoted by $\FO$. 
Let $\tau$ be a relational vocabulary and $R\in \tau$ be a relation symbol of arity $r$.
An atomic formula is a formula of the form $x = y$ or $R(x_1,\dots, x_r)$. 
A literal is an atomic or a negated atomic formula. 
A quantifier-free formula is a formula that contains no quantifiers and a formula is in negation normal form (NNF) if the negation symbols occurs only front of atoms.
A formula $\varphi$ is in prenex normal form if $\varphi$ has the form $Q_1x_1\dots Q_nx_n \psi$, where $\psi$ is quantifier free and $Q_1,\dots,Q_n\in\{\exists, \forall\}$.
The classes $\SigForm 0$ and $\PiForm 0$ both consist of quantifier free formulas.
Then, for $t\geq 0 $, the class $\SigForm{t+1}$ includes all formulas of the form $\exists x_1\dots \exists x_\ell \varphi$, where $\varphi\in \PiForm{t}$.
Similarly, $\PiForm{t+1}$ includes all formulas of the form $\forall x_1\dots \forall x_\ell \varphi$, where $\varphi\in \SigForm{t}$.

\subparagraph{Fagin Definability}
The first-order variables range over individual elements of the universe.
In second-order logic, one also quantifies relation variables which range over relations on the universe.
We now introduce first-order formulas where we also allow relation variables.  
%Let $\varphi(x_1,\dots,x_n,X_1,\dots,X_\ell)$ be a $\FO$-formula with free first order variables $x_1,\dots,x_n$ and free relation variables $X_1,\dots,X_\ell$.
Let $\tau$ be a vocabulary, $X_i$ for $i\leq n$ be free relation variables of arity $s_i$ and $\varphi(X_1,\dots, X_n)$ be a $\FO$-formula in $\tau$.
Moreover, let $\calA$ be a $\tau$-structure and $S_i\subseteq A^{s_i}$ be relations over $\calA$ for $i\leq n$.
Then we say that the tuple $\bar S = (S_1,\dots,S_n)$ is a solution for $\varphi$ in $\calA$ if $\calA \models \varphi(\bar S)$.
We call the following decision problem, the problem \emph{Fagin-defined} by $\varphi$. 

\problemdef{$\FD_\varphi$ --- Fagin definability for fixed $\varphi\in\FO$}{A $\tau$-structure $\calA$}{Is there a solution for $\varphi$ in $\calA$}

Let $\Fragment\subseteq\FO $ be a class of formulas, then by $\FD\text-\Fragment$ we denote the class of all problems $\FD_\varphi$ such that $\varphi\in \Fragment$.
The following result regarding $\FO$ is known.
	\begin{proposition}[{\cite[Cor. 4.35]{DBLP:series/txtcs/FlumG06}}]
		$\NP = \FD\text-\FO = \FD\text-\PiForm 2$.
	\end{proposition}

Next we introduce the following \emph{weighted} version of Fagin definabilty, where we restrict our solution to have a specific size for a single free relation symbol $S$ of arity $s$.
\problemdef{$\WD_\varphi$ --- weighted Fagin definability for fixed $\varphi\in\FO$}{A $\tau$-structure $\calA$ and $k\in \mathbb N$}{Is there a solution for $\varphi$ of cardinality $k$}
As before, for a class $\Fragment\subseteq\FO $ of formulas, we denote by $\WD\text-\Fragment$ the class of all problems $\WD_\varphi$ such that $\varphi\in \Fragment$.

\begin{example}\label{Ex:fd}
	The problem $\Clique$ is defined as follows. Given a graph $\calG\dfn (V,E)$ and $k\in \mathbb N$. Is there a set $S\subseteq V$ such that $|S|=k$ and $(u,v)\in E$ for every $x,y\in S$?
	Then $\Clique$ is $\WD_{\varphi_c}$, where 
	\[
	\varphi_c(X) \dfn \forall x\forall y \bigl((X(x)\land X(y)\land x\not=y)\rightarrow E(xy)\bigr).  
	\]
	Consequently, $\Clique$ is in $\WD\text-\PiForm 1$.
	
	Moreover, Let $\DomSet$ be the problem to determine if a graph $\calG$ contains a set $S\subseteq V$ such that $|S|=k$ and every vertex in $V\setminus S$ is incident to some vertex in $S$?
	Then $\DomSet$ is in $\WD\text-\PiForm 2$ since the problem is $\WD_{\varphi_d}$, where
	\[
		\varphi_d(X) \dfn \forall x\exists y \bigl(X(y)\land (E(x,y) \lor x=y)\bigr).  
	\]
\end{example}

\subparagraph*{Parameterized Complexity Theory}
A \emph{parameterized problem} (PP) $P\subseteq\Sigma^*\times\mathbb N$ is a subset of the crossproduct of an alphabet and the natural numbers.
For an \emph{instance} $(x,k)\in\Sigma^*\times\mathbb N$, $k$ is called the (value of the) \emph{parameter}.
A \emph{parameterization} is a polynomial-time computable function that maps a value from $x\in\Sigma^*$ to its corresponding $k\in\mathbb N$.
The problem $P$ is said to be \emph{fixed-parameter tractable} (or in the class $\FPT$) if there exists a deterministic algorithm $\mathcal A$ and a computable function $f$ such that for all $(x,k)\in\Sigma^*\times \mathbb N$, algorithm $\mathcal A$ correctly decides the membership of $(x,k)\in P$ and runs in time $f(k)\cdot|x|^{O(1)}$.
The problem $P$ belongs to the class $\XP$ if $\mathcal A$ runs in time $|x|^{f(k)}$ on a deterministic machine. %, whereas $\XNP$ is the non-deterministic counterpart of $\XP$.
Abusing a little bit of notation, we write $\mathcal C$-machine for the type of machines that decide languages in the class $\mathcal C$, and we will say a function $f$ is \emph{$\mathcal C$-computable} if it can be computed by a machine on which the resource bounds of the class $\mathcal C$ are imposed.
The class $\para\NP$ includes problems decidable by a nondeterministic algorithm $\mathcal A$ which runs in time $f(k)\cdot|x|^{O(1)}$ for some computable function $f$. 
One can define a parameterized complexity class $\para\mathcal C$ corresponding to a complexity class $\mathcal C$ via a \emph{precomputation on the parameter}.
\begin{definition}
	Let $\mathcal C$ be any complexity class.
	Then $\para\mathcal C$ is the class of all PPs $P\subseteq\Sigma^*\times\mathbb N$ such that there exists a computable function $\pi\colon\mathbb N\to\Delta^*$ and a language $L\in\mathcal C$ with $L\subseteq\Sigma^*\times\Delta^*$ such that for all $(x,k)\in\Sigma^*\times\mathbb N$ we have that $(x,k)\in P \Leftrightarrow (x,\pi(k))\in L$.
\end{definition}
Notice that $\para\complClFont{P}=\FPT$ and the two definitions of $\para\NP$ are equivalent.
%The complexity classes $\mathcal{C}\in\{\NP,\PSPACE,\NEXP\}$ are used in the $\para\mathcal C$ context by us.

A problem $P$ is in the complexity class $\WP$, if it can be decided by a NTM running in time $f(k)\cdot|x|^{O(1)}$ steps, with at most $g(k)$-many non-deterministic steps, where $f,g$ are computable functions.
Moreover, $\WP$ is contained in the intersection of $\para\NP$ and $\XP$ (for details see the textbook of Flum and Grohe~\cite{DBLP:series/txtcs/FlumG06}). 

Let $c\in\mathbb N$ and $P\subseteq\Sigma^*\times\mathbb N$ be a PP, then the \emph{$c$-slice of $P$}, written as $P_c$ is defined as $P_c\coloneqq\{\,(x,k)\in\Sigma^*\times\mathbb N\mid k=c\,\}$.
Notice that $P_c$ is a classical problem then.

\begin{definition}\label{def:fpt-reduction}
	Let $P\subseteq\Sigma^*\times\mathbb N,Q\subseteq\Gamma^*$ be two PPs.
	One says that $P$ is \emph{fpt-reducible} to $Q$, $P\fptreduction Q$, if there exists an $\FPT$-computable function $f\colon\Sigma^*\times\mathbb N\to\Gamma^*\times\mathbb N$ such that
	\begin{itemize}
		\item for all $(x,k)\in\Sigma^*\times\mathbb N$ we have that $(x,k)\in P\Leftrightarrow f(x,k)\in Q$,
		\item there exists a computable function $g\colon\mathbb N\to\mathbb N$ such that for all $(x,k)\in\Sigma^*\times\mathbb N$ and $f(x,k)=(x',k')$ we have that $k'\leq g(k)$.
	\end{itemize}
\end{definition}
Finally, in order to show that a problem $P$ is $\para\mathcal C$-hard (for some complexity class $\mathcal C$) it is sufficient to prove that for some $c\in \mathbb N$, the slice $P_c$ is $\mathcal C$-hard in the classical setting.

To define the complexity classes in $\W{}$-hierarchy, the parameterized version of the problem $\WD_\varphi$ is now defined as follows.
\paraproblemdef{$\p\WD_\varphi$ --- parameterized weighted Fagin definability for fixed $\varphi\in\FO$}{A $\tau$-structure $\calA$ and $k\in \mathbb N$}{$k$}{Is there a solution for $\varphi$ of cardinality $k$}

The complexity classes of the $\W{}$-hierarchy are characterized via the following definition.
\begin{definition}[{\cite[Def. 5.1]{DBLP:series/txtcs/FlumG06}}]\label{def:wt}
	For every $t\geq 1$, we let $\W t\dfn [\p\WD\text-\PiForm t]^{\FPT}$. 
	The class $\W t$ forms the $t$-th level of the $\W{}$-hierarchy.
\end{definition} 

Alternatively, the $\W{}$-hierarchy can be defined via the weighted satisfiability problem for propositional formulas. 
%We will need the following characterization for proving Theorem~\cite{bla}.
%
Let $I$ be a non-empty index set and $d\in\mathbb N$. 
Consider the following special subclasses of propositional formulas:
\[\begin{array}{@{}r@{\,}c@{\,}l@{}}
	\Gamma_{0, d} & =&  \{\ell_1\land\dots\land \ell_c \mid \ell_1,\dots, \ell_c \text{ are literals and }c\leq d \},\\
	\Delta_{0, d} & = & \{\ell_1\lor\dots\lor \ell_c \mid \ell_1,\dots, \ell_c \text{ are literals and }c\leq d \},\\
	\Gamma_{t, d} & = & \left\{\,\bigwedge\limits_{i\in I} \alpha_i \,\middle|\, \alpha_i \in \Delta_{t-1,d}  \text{ for } i \in I\, \right\}, \\
	\Delta_{t, d} &= & \left\{\,\bigvee\limits_{i\in I} \alpha_i \,\middle|\, \alpha_i \in \Gamma_{t-1,d}  \text{ for } i \in I\, \right\}.
\end{array}\]
Finally, $\Gamma^{+}_{t,d}$ (resp. $\Gamma^{-}_{t,d}$) denote the class of all positive (negative) formulas in $\Gamma_{t,d}$.

The parameterized weighted satisfiability problem ($\WSAT$) for propositional formulas is defined as below.
\paraproblemdef{$\p\wsat{t,d}$ --- parameterized weighted satisfiability}{a $\Gamma_{t,d}$-formula $\alpha$ with $t, d \geq 1 $ and $k\in \mathbb{N}$}{$k$}{is there a satisfying assignment for $\alpha$ of weight $k$}

The classes of the $\complClFont{W}$-hierarchy are defined equivalently in terms of these problems.
\begin{proposition}[{\cite[Thm.~7.1]{DBLP:series/txtcs/FlumG06}}]\label{theorem-wsat}
	For every $t\geq 1$ the following problems are $\W t$-complete under fpt-reductions.
	\begin{itemize}
		\item $\p\wsatpos{t,1}{}$ if $t$ is even and  $\p\wsatneg{t,1}{}$ if $t$ is odd.
		\item $\p\wsat{t,d}{}$ for every $t,d\geq 1$.
	\end{itemize}
\end{proposition}

\begin{figure}
	\centering
		\resizebox{\linewidth}{!}{\begin{tikzpicture}[every node/.style={scale=0.8,opacity=1,font=\fontsize{15}{25}\selectfont},fill=gray, fill opacity=0.1]
% \draw[step=1cm,gray,very thin] (-10,0) grid (10,12);
% \node at (0,0) {\textbf{(0,0)}};

\fill[fill=white,fill opacity=1] (-5.7,0) parabola bend (0,10) (5.7,0);
%% Horizontal bar
\draw[very thick] (7,0) -- (-7,0);

\fill (-1.7,0) parabola bend (0,3.5) (1.7,0);
\fill (-2.2,0) parabola bend (0,4.5) (2.2,0);
\fill (-2.7,0) parabola bend (0,5.5) (2.7,0);
\fill (-3.2,0) parabola bend (0,6.5) (3.2,0);
\fill (-3.7,0) parabola bend (0,7.5) (3.7,0);
\fill (-5.7,0) parabola bend (0,10) (5.7,0);

% FPT
\draw (-1.7,0) parabola bend (0,3.5) (1.7,0);
% W1
\draw (-2.2,0) parabola bend (0,4.5) (2.2,0);
% W2
\draw (-2.7,0) parabola bend (0,5.5) (2.7,0);
% . . .
\draw (-3.2,0) parabola bend (0,6.5) (3.2,0);
% WP
\draw (-3.7,0) parabola bend (0,7.5) (3.7,0);
% para-NP
\draw[thick] (-5.7,0) parabola bend (0,10) (5.7,0);

\node at (0,1.7) {$\FPT$};
\node at (0,3.85) {$\W{1}$};
\node at (0,4.9)  {$\W{2}$};
\node at (0,6.15) {$\vdots$};
\node at (0,7)    {$\W\Ptime$};
\node at (0,9)    {$\para\NP$};
 
\node[anchor=west] at (5.5,1.5)  (fptr){p-\scshape VertexCover};
%\node[anchor=west] at (5.5,1)    {|Sol|};

\node[anchor=west] at (5.5,3.85) (w1r){p-\scshape Clique};
\node[anchor=west] at (5.5,3.35) {p-\scshape IndependentSet};

\node[anchor=west] at (5.5,4.9) (w2r){p-\scshape DominatingSet};
%\node[anchor=west] at (5.5,4.4) {|Sol|};

\node[anchor=west] at (5.5,9) (pnpr){p-\scshape Colouring};
%\node[anchor=west] at (5.5,8.5) {\#colours};

\draw[thick,o-]  (.5,1.5) -- (fptr) ;
\draw[thick,o-]  (.5,3.85) -- (w1r) ;
\draw[thick,o-]  (.5,4.9) -- (w2r) ;
%\draw[o-]  (.5,7) -- (wpr) ;
\draw[thick,o-]  (1,9) -- (pnpr) ;

\node[anchor=west] at (-9,1.5)  (fptl){p-WSAT$(\Gamma_{1,d}^{+})$};
%\node[anchor=west] at (-9,.9)    { weight(Assignment)};

\node[anchor=west] at (-9,4.92) (w2l){ p-WSAT$(\Gamma_{t,d})$};
%\node[anchor=west] at (-9,4.4) { weight(Assignment)};

\node[anchor=west] at (-9,7) (wpl){ p-WSAT(CIRC)};
%\node[anchor=west] at (-9,6.5) { weight(Assignment)};

%\node[anchor=west] at (-9,9) (pnpl){ SAT(CNF)};
%\node[anchor=west] at (-9,8.5) { \#literals/clause};

\draw[thick,o-]  (-.5,1.5) -- (fptl) ;
\draw[thick,-]  (-4.6,4.92) -- (w2l) ;
\draw[thick,o-]  (-.5,7) -- (wpl) ;
%\draw[thick,o-]  (-1,9) -- (pnpl) ;

\draw[thick,decoration={brace, mirror, amplitude=4pt},decorate,thick] (-4.4,6) -- node[rotate=90,left=8pt,anchor=south] {} (-4.4,3.85);
%%%%%%%%%
\draw[thick,-o,dashed]  (-4.3,6) -- (-.5,6) ;
\draw[thick,-o,dashed] (-4.3,3.85) -- (-.5,3.85) ;
\end{tikzpicture}}
\caption{Landscape of relevant parameterized complexity classes with complete problems. The definition of several of these complete problems are mentioned in the relevant proofs.}\label{fig:pc-landscape}
\end{figure}
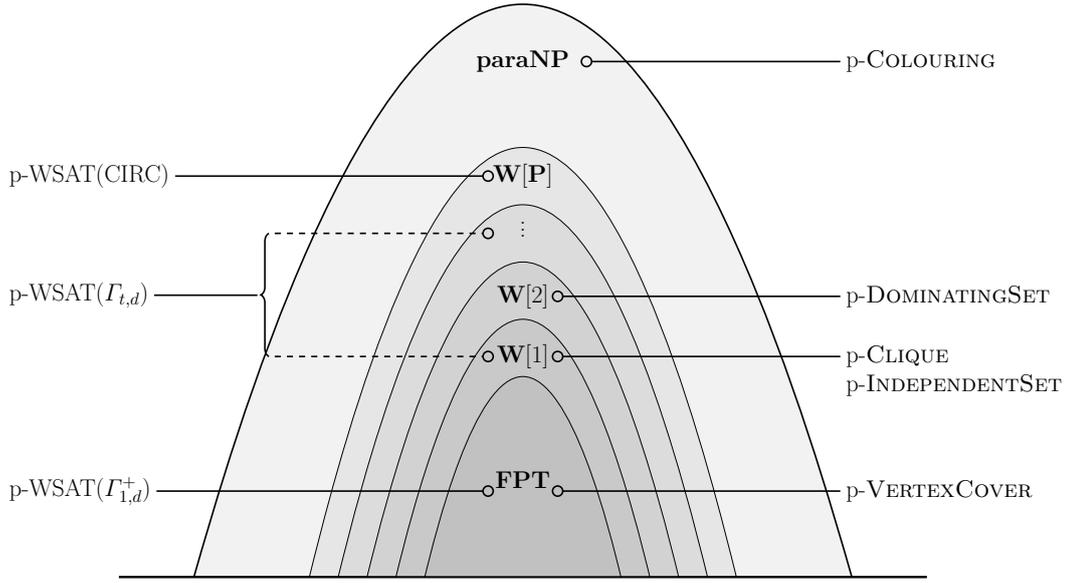
Figure~\ref{fig:pc-landscape} draws the complexity landscape with complete problems in parameterized complexity that are relevant.

\subparagraph*{Team-based Logics}
We assume basic familiarity with predicate logic~\cite{DBLP:books/daglib/0082516}. 
We consider first-order vocabularies $\tau$ that are sets of \emph{function} symbols and \emph{relation} symbols with an equality symbol $=$. 
Let $\VAR$ be a countably infinite set of \emph{first-order variables}.
Terms over $\tau$ are defined in the usual way, and the set of well-formed formulas of first-order logic ($\FO$) is defined by the following EBNF:
\[
	\psi \Coloneqq
	t_1 =t_2\mid 
	R(t_1,\dots,t_k)\mid
	\lnot R(t_1,\dots,t_k)\mid
%	f(t_1,\dots,t_k)\mid 
	\psi\land\psi\mid
	\psi\lor\psi\mid
	\exists x\psi\mid
	\forall x\psi,
\]
where $t_i$ are terms $1\leq i\leq k$, $R$ is a $k$-ary relation symbol from $\sigma$, $k\in\mathbb N$, and $x\in\VAR$.
If $\psi$ is a formula, then we use $\VAR(\psi)$ for its set of variables, and $\Fr(\psi)$ for its set of \emph{free variables}.
We evaluate $\FO$-formulas in $\tau$-structures, which are pairs of the form $\calA=(A,\tau^\calA)$, where $A$ is the \emph{domain} of $\calA$ (when clear from the context, we write $A$ instead of $\dom(\calA)$), and $\tau^\calA$ interprets the function and relational symbols in the usual way (e.g., $t^\calA\langle s\rangle=s(x)$ if $t=x\in\VAR$).
If $\mathbf t=(t_1,\dots,t_n)$ is a tuple of terms for $n\in\mathbb N$, then we write $\mathbf t^\calA\langle s\rangle$ for $(t_1^\calA\langle s\rangle, \dots, t_n^\calA\langle s\rangle)$.

Dependence logic $\D$ extends $\FO$ by dependence atoms of the form $\depa{t}{u}$ where $\mathbf t$ and $\mathbf u$ are tuples of terms.
Inclusion logic $\INC$ in obtained by adding to $\FO$ the inclusion atoms of the form $\inca{t}{u}$ for tuples $\mathbf t$ and $ \mathbf u$ of terms.
Finally, independence logic $\IND$ extends $\FO$ by independence atoms of the form $\indepa{t}{u}{v}$ for tuples $\mathbf t, \mathbf u$ and $\mathbf v$ of terms.
We call expressions of the kind $t_1=t_2, R(\mathbf t),\depa{t}{u},\inca{t}{u}$ and $\indepa{t}{u}{v}$ \emph{atomic formulas}.

The semantics is defined through the concept of a team.
Let $\calA$ be a structure and $X\subseteq\VAR$, then an \emph{assignment} $s$ is a mapping $s\colon X\rightarrow A$. 
\begin{definition}
Let $X\subseteq\VAR$. A \emph{team $T$ in $\calA$ with domain $X$} is a set of assignments $s\colon X\to A$.	
\end{definition}
For a team $T$ with domain $X\supseteq Y$ define its \emph{restriction} to $Y$ as $T{\upharpoonright} Y\coloneqq\{\,s{\upharpoonright} Y \mid s\in T\,\}$.
If $s\colon X\to A$ is an assignment and $x\in\VAR$ is a variable, then $s^x_a\colon X\cup\{x\}\to A$ is the assignment that maps $x$ to $a$ and $y\in X\setminus\{x\}$ to $s(y)$. 
Let $T$ be a team in $\calA$ with domain $X$. 
Then any function $f\colon T\to \mathcal{P}(A)\setminus\{\emptyset\}$ can be used as a \emph{supplementing function} of $T$ to extend or modify $T$ to the \emph{supplemented team} $T^x_f\coloneqq\{\,s^x_a\mid s\in T,a\in f(s)\,\}$. 
For the case $f(s)=A$ is the constant function we simply write $T^x_\calA$ for $T^x_f$. 
The semantics of formulas is defined as follows.
\begin{definition}\label{def-semantics}
	Let $\tau$ be a vocabulary, $\calA$ be a $\tau$-structure and $T$ be a team over $\calA$ with domain $X\subseteq\VAR$. Then,
\begin{alignat*}{3}
	& \calA,T\models t_1=t_2 && \;\text{ iff }\; && \forall s\in T: t_1^\calA\langle s\rangle=t_2^\calA\langle s\rangle,\\
	& \calA,T\models t_1\neq t_2 && \;\text{ iff }\; && \forall s\in T: t_1^\calA\langle s\rangle\neq t_2^\calA\langle s\rangle,\\
	& \calA,T\models R(t_1,\dots,t_n) && \;\text{ iff }\;  && \forall s\in T: (t_1^\calA\langle s\rangle,\dots,t_n^\calA\langle s\rangle)\in R^{\calA},\\
	& \calA,T\models \neg R(t_1,\dots,t_n) && \;\text{ iff }\;  && \forall s\in T:  (t_1^\calA\langle s\rangle,\dots,t_n^\calA\langle s\rangle)\not\in R^{\calA},\\
	& \calA,T\models \depa{t}{u} && \;\text{ iff }\;  && \forall s_1, s_2\in T: \mathbf t^\calA\langle s_1\rangle=\mathbf t^\calA\langle s_2\rangle \implies  \mathbf u^\calA\langle s_1\rangle=\mathbf u^\calA\langle s_2\rangle ,\\		
	& \calA,T\models \inca{t}{u} && \;\text{ iff }\;  && \forall s_1\in T, \exists s_2\in T: \mathbf t^\calA\langle s_1\rangle=\mathbf u^\calA\langle s_2\rangle, \\
	& \calA,T\models \indepa{t}{u}{v} && \;\text{ iff }\;  && \forall s_1, s_2\in T: \mathbf v^\calA\langle s_1\rangle=\mathbf v^\calA\langle s_2\rangle \text{ then } \exists s_3\in T: \\
	& &&   && \mathbf {vt}^\calA\langle s_3\rangle=\mathbf {vt}^\calA\langle s_1\rangle \text{ and } \mathbf {u}^\calA\langle s_3\rangle=\mathbf {u}^\calA\langle s_2\rangle,\\		
	& \calA,T\models \varphi_0\land \varphi_1  && \;\text{ iff }\;  && \calA,T\models \varphi_0 \text{ and } \calA,T\models \varphi_1, \\
	& \calA,T\models \varphi_0\lor \varphi_1  && \;\text{ iff }\;  && \exists T_0\exists T_1: T_0\cup T_1=T  \text{ and } \calA,T_i\models \varphi_i  \, \text{ for }i =0,1,\\
	& \calA,T\models\exists x\varphi && \;\text{ iff }\;  && \calA,T^x_f\models\varphi\text{ for some }f\colon T\to \mathcal{P}(A)\setminus\{\emptyset\},\\
	& \calA,T\models\forall x\varphi && \;\text{ iff }\;  && \calA,T^x_\calA\models\varphi.
\end{alignat*}
\end{definition}

For a structure $\calA$ and a team $T$ over $X$ in $\calA$, we let $\rel(T)$ denote the relation \emph{defined} by $T$.
That is, $\rel(T)\dfn \{\,\tuple{a}\mid s(\tuple{x})=\tuple{a},s\in T \,\}$.
Moreover, we say that a formula $\varphi$ is \emph{flat} if for any team $T$ over $\Fr(\varphi)$ we have that $\calA,T\models \varphi$ if and only if $\calA,\{s\}\models \varphi$ for every $s\in T$.
The $\FO$-formulas satisfy this flatness property.
Notice that, for $\FO$-formulas, by singleton equivalence, team semantics and classical Tarski semantics coincide, i.e., $\calA,\{s\}\models \varphi$ if and only if $\calA\models_s \varphi$.
%Notice that we only consider formulas in negation normal form (NNF).
%In the team semantics setting, disjunction and existential quantifier are given two different meanings.
%The above defined semantics is the so-called \emph{lax}-semantics, whereas an alternative is the \emph{strict}-semantics.
%In strict-semantics, the split of teams have to be disjoint and the supplementing function is replaced by a function  $f\colon T\to A$.
%That is, the function $f$ assigns a single element $a\in A$ to each $s\in T$.
%For dependence logic, the two semantics coincide due to the downwards closure property. 
%That is, for any $\D$-formula $\varphi$, if $\calA,T\models \varphi$ then $\calA,P\models \varphi$ for every $P\subseteq T$. 
%For this reason we only consider lax semantics for $\D$.
Furthermore, note that $\calA,T\models \varphi$ for all $\varphi$ when $T=\emptyset$ (this is also called the \emph{empty team property}).
Finally, $\calC$-formulas for every $\calC\in \{\D,\INC,\IND\}$ are \emph{local}, that is, for a team $T$ in $\calA$ over domain $X$ and a $\D$-formula $\varphi$, we have that $\calA,T\models\varphi$ if and only if $\calA,T{\upharpoonright}{\Fr(\varphi)}\models\varphi$.
%$\INC$ and $\IND$-formulas are also local under lax-semantics but not under strict-semantics~\cite{DBLP:journals/apal/Galliani12}.
%Notice that strict-semantics is relatively stricter (as the name suggest) than the lax-semantics \cite{DBLP:journals/apal/Galliani12}. 
%That is, for every $\INC$-formula $\varphi$, if $\calA,T\models_s \varphi$ then  $\calA,T\models_\ell \varphi$, where the subscript $s$ and $\ell$ indicates the choice of the semantics. 

We now extend the formulas classes ($\SigForm{t}$ and $\PiForm{t}$) to the logics under consideration. 
To this end, $\DPi{t}\subseteq\D$ (resp., $\DSigma{t}$) denotes the collection of formulas $\varphi$ of the form $\varphi\dfn Q_1x_1Q_2x_2\dots Q_tx_t\psi$ such that $\psi$ is a quantifier free $\D$-formula, $Q_i\in \{\forall,\exists\}$ and $Q_1=\forall$ ($Q_1=\exists$). 
In other words, $\varphi$ is a $\D$-formula that starts with a $\forall$-quantifier (resp., $\exists$) and has $t$-alternations of quantifiers.
The classes $\INCPi{t}\subseteq\INC$ (resp., $\INCSigma{t}$) for $\INC$ and $\INDPi{t}\subseteq\D$ (resp., $\INDSigma{t}$) for $\IND$ are similarly defined.

\subparagraph{Weighted Team Definability}
Now we introduce a novel version of the weighted definability problem for formulas in team-based logics.
Let $\calC\in \{\D,\INC,\IND\}$, $\varphi$ be a fixed $\calC$-formula over free variables $\Vars(\varphi)$ and $k\in \mathbb N$. 
Then given a structure $\calA$, the \emph{weighted-team} definable problem $\WTD_\varphi$ asks if there is a team of size $k$ for $\varphi$ over $\Vars(\varphi)$ in $\calA$.
\problemdef{$\WTD_\varphi$ --- weighted team definability for fixed $\varphi$}{A $\tau$-structure $\calA$ and $k\in \mathbb N$}{Is there a team $T$ over $\Vars(\varphi)$ such that $|T|=k$ and $\calA,T\models \varphi$}
%Recall that if $\varphi$ is a sentence then for any structure $\calA$ and a team $T$, we have that $(\calA, T)\models \varphi$ iff  for all $T'$: $\calA, T'\models \varphi$.
	%In other words, there is a satisfying team $T$ for $\varphi$ in $\calA$ if and only if $(\calA, \emptyset)\models \varphi$.
	%\todo[inline]{What kind of complexity results do we get for sentences?}
Then the analogous parameterized version of $\WTD_\varphi$ is defined as follows.
\paraproblemdef{$\p \WTD_\varphi$ --- parameterized weighted team definability for fixed $\varphi$}{A $\tau$-structure $\calA$ and $k\in \mathbb N$}{$k$}{Is there a team $T$ over $\Vars(\varphi)$ such that $|T|=k$ and $\calA,T\models \varphi$}

Note that  the problem  $\WTD_\varphi$ references the set of free variables $\Vars(\varphi)$ of the formula $\varphi$. 
As a consequence, our parameterization is trivial for sentences since there are only two teams $\emptyset$ and $\{\emptyset\}$ with the empty team domain.
As before, for a set $\Fragment\subseteq \calC$ of formulas, we denote by $\WTD\text-\Fragment$ the class of problems $\WTD_\varphi$ such that $\varphi\in \Fragment$.

\section{Complexity Results for Weighted Team Definability}
\subsection{First-Order Formulas}

We begin our study of the complexity for $\p \WTD_\varphi$ in the case $\varphi $ is a pure $\FO$-formula under team semantics.
Notice that the consequence of disallowing free relation variables in $\varphi$ is that $\p\WTD_\varphi$ is different than the weighted Fagin definability $\p\WD_\varphi$.
The following theorem establishes that the two problems are also different from the classical complexity theoretic point of view. 
Here, we assume basic familiarity about the circuit complexity classes $\mathrm{TC}^0$ and $\mathrm{AC}^0$ (for an introduction into this area, see the textbook of Vollmer~\cite{DBLP:books/daglib/0097931}).

\begin{theorem}\label{thm:WT-FO}
For any $\FO$-formula  $\varphi$ the problem $\WTD_\varphi$ is in DLOGTIME-uniform $\mathrm{TC}^0$.
\end{theorem}
\begin{proof} 
The proof uses the flatness property of $\FO$-formulas under team semantics: 
\[
\calA,T\models \varphi \Leftrightarrow \forall  s \in T: \ \calA\models_s \varphi.
\]
It is well know that $\calA\models_s \varphi$ can be decided by $\mathrm{AC}^0$-circuits, whence the original question reduces to counting the number $t$ of satisfying assignments of $\varphi$ and checking whether $t\ge k$. 
This can be easily simulated by DLOGTIME-uniform  $\mathrm{TC}^0$ circuits as we can hardcode all possible assignments into the circuit. 
Here, notice that $\varphi$ is fixed and thereby the number of free variables are fixed to some constant $c\in\mathbb N$.
Then, the input is the structure $\mathcal A$ of size $n$ yielding $O(n^c)$ many assignments.
\end{proof}

\subsection{Inclusion Logic}
%Capturing $\W{}$-hierarchy, and the complexity class $\WP$ by Juha's Idea. 
In this section, we relate the $\W{}$-hierachy and $\W\Ptime$ to weighted team definability for inclusion logic formulas. 
	First observe that if $\varphi$ is an $\INC$-sentence, then the problem $\p \WTD_\varphi$ is in $\FPT$.
	This is due to the reason that the data complexity of fixed $\INC$-sentences is in $\Ptime$~\cite{gallhella13}.

	\begin{theorem}\label{inc:sentences}
		Let $\varphi$ be an $\INC$-sentence, then $\p\WTD_{\varphi}$ is in $\FPT$.
	\end{theorem}
 	\begin{proof}
 	Recall that an $\INC$-sentence $\varphi$ has a satisfying team $T$ in $\calA$ if and only if $\calA,\{\emptyset\}\models \varphi$.
 	Then $\varphi$ is true in $\calA$ if and only if there is a team $T$ such that $|T|=1$ and $\calA,T\models \varphi$.
 	\end{proof}

Now we prove, that $\p\WTD_\varphi$ can already be $\W1$-hard when $\varphi$ is a quantifier-free $\FO(\subseteq)$-formula with free variables.

\begin{theorem}\label{thm:WT-INC-W1}
 There is a quantifier-free $\FO(\subseteq)$-formula $\varphi$ such that the problem $\p\WTD_\varphi$ is $\W1$-hard and in $\W2$. 
\end{theorem}
\begin{proof}
We present a reduction from the $\W 1$-complete problem $\p\Clique$ to $\p\WTD_\varphi$ such that $\varphi$ is a quantifier free $\INC$-formula.
Let $G\dfn (V,E)$ be a graph and $k\in \mathbb N$.
Then, we let $\varphi\coloneqq E(x,y)\wedge x\neq y\wedge y\subseteq x \wedge x\subseteq y$.
We claim that $G$ has a clique of size $k$ if and only if $G,T\models \varphi$ for a team $T$ of size $(k^2-k)$. 
It is straightforward to check that the existence of a $k$-clique is equivalent to $\varphi $ having a satisfying team of cardinality $k(k-1)$ with exactly the same values for $x$ and $y$. 

For containment in $\W2$, it suffices to note that the formula $\varphi$ can be expressed as an $\FO$-sentence $\psi(S)$ with a $\forall\exists$-quantifier prefix where the auxiliary binary predicate $S$ encodes the team $T$.
This gives an $\FPT$-reduction between $\p\WTD_\varphi$ and $\p\WD_\psi$.
The result follows since $\W2\dfn[\p\WD\text-\PiForm{2}]$.
\end{proof}

This result can be strengthened to more general formulas as witnessed by the following corollary.
\begin{corollary}\label{cor:inc-lor}
	For any quantifier-free $\FO(\subseteq)$-formula $\varphi$ without $\lor$, the problem $\p\WTD_\varphi$ is $\W1$-hard and in $\W2$. 
\end{corollary}
\begin{proof}
For containment in $\W2$, it suffices to note that the any quantifier-free formula without disjunction can be expressed as an $\FO$-sentence $\psi(S)$ with a $\forall\exists$-quantifier prefix where the auxiliary binary predicate $S$ encodes the team $T$.	
\end{proof}

\begin{theorem}\label{thm:WT-INC-W2}
There is an $\FO(\subseteq)$-formula $\varphi$ with $\forall\exists$-quantifier prefix for which 
the problem $\p\WTD_\varphi$ is $\W2$-complete.
\end{theorem}

\begin{proof}
We present a reduction from the $\W2$-complete problem $\p\DomSet$ to $\p\WTD_\varphi$ such that $\varphi$ is a $\INC$-formula with $\forall\exists$-quantifier prefix.
Let $G\dfn (V,E)$ be a graph and $k\in \mathbb N$.
Then we let, $\varphi\dfn \forall x \exists y (y\subseteq z \wedge ( E(x,y)\vee x=y))$.
It is straightforward to check that $G$ has a dominating set of size $k$ if and only if $G,T\models \varphi$ for a team $T$ with domain $\{z\}$  of size $k$.

For $\W 2$-membership, notice that for all graphs $G$ and teams $T$:
\[
G,T\models \varphi \Leftrightarrow (G, \rel(T))\models \varphi_d(X),
\]
where $\varphi_d(X)$ is the first-order sentence encoding the problem $\DomSet$ (see Example~\ref{Ex:fd}). 
A formal proof for the above equivalence is similar to the one given in Theorem~\ref{D:W1}.
\end{proof}

The next lemma sets the stage for generalizing the two previous theorems to arbitrary levels of the $\W{}$-hierarchy. 
To formulate the result, we assume an encoding of a formula $\psi \in \Gamma^{+}_{t,d}$ (and a truth assignment) by its syntax circuit $A_{\psi}=(A,E,I,o)$, where $A$ is the set of subformulas of $\psi$, $E$ is the immediate subformula relation, $I\subseteq A$ are the variables of $\psi$, $o$ is a constant symbol interpreted by $\psi$. 
Finally a free relation variable $S\subseteq I$ can be used to  represent a truth assignment for the variables. Note that our encoding of $\psi$ works for any $t\in \mathbb N$ but for the definability result below $t$ has to be fixed.

\begin{lemma}\label{INC:WT-level}
Let $t\in \mathbb{N}$. Then there exists a fixed formula $\varphi_t\in \FO(\subseteq)$ with one free variable $z$ such that for all  $\psi \in \Gamma^{+}_{t,d}$ and $k\ge 1$: $\psi$ has a satisfying assignment of weight $k$ if and only if $A_{\psi},T\models \varphi_t$, for some team $T$ of cardinality $k$.
\end{lemma}

\begin{proof} 
%\todo[inline]{Use myolemmapic lemma 16/prop 20 from the source and explain in few sentences what to do.}
Without loss of generality, we assume $d=1$. 
For higher $d$-values, the presented proof easily generalizes via a conjunction/disjunction of arity $d$. 
By the results of Galliani and Hella~\cite{gallhella13}, it suffices to show that the required formula can be expressed by a first-order sentence $\theta(S)$ in which the relation symbol $S$ occurs only postively. Then  the existence of  $\varphi_t(z)$ satisfying 
 \begin{equation}\label{myopic}
 A_{\psi},T\models \varphi_t \Leftrightarrow  A_{\psi}\models \theta(S),
 \end{equation}
 for all non-empty $T$ and $\rel(T)=S$ follows. 
 Note that $\theta(S)$ is not true under the  assignment  setting all the variables to false, but on the other hand $\varphi_t$ is always satisfied for $T=\emptyset$  by the empty team property.  
It is easy to check that $\theta(S)$ can be expressed as follows:
 \begin{equation*}\theta(S)\coloneqq \forall x_1\big(\neg E(o,x_1) \vee \exists x_2(E(x_1,x_2)\wedge \cdots Qx_t(E(x_{t-1},x_t)\wedge I(x_t)\wedge S(x_t)  )\cdots \big).
  \end{equation*}
The relation symbol $S$ has only one occurrence in the formula and it is positive. 
Now by Proposition 20 of \cite{gallhella13}, there exists a formula $\varphi_t$ such \eqref{myopic} holds for the sentence $\forall \vec{x}(S(\vec{x})\rightarrow \theta(S))$ for all $\calA$ and all $T$. 
It is easy to see that $\theta(S)$ is equivalent with $\forall \vec{x}(S(\vec{x})\rightarrow \theta(S))$ modulo the cases when $S=\emptyset$.
In fact, it is straightforward to show that $\varphi_t$ can be obtained from $\theta(S)$ simply by replacing $S(x_t)$ by the inclusion atom $x_t\subseteq z$. 
The proof then is analogous to the proof of Theorem~\ref{D:W1}.
\end{proof}

Notice further that the translation of the formula~$\theta$ to an $\INC$-formula only introduces inclusion atoms and, in particular, does not require any further quantification. 
Therefore, the following corollary follows immediately from the proof in Lemma~\ref{INC:WT-level}.

\begin{corollary}\label{cor:inc-wt}
Let $t\geq 2$ be even. 
Then there is an $\FO(\subseteq)\text-\PiForm{t}$-formula ${\varphi_t}$ for which 
the problem $\p\WTD_{\varphi_t}$ is $\W t$-complete.
Moreover, $\W t \subseteq [\p\WTD\text-\INCPi t]^{\FPT}$ for all even $t\geq 1$ and 
 $\bigcup_{t\ge1}\W t \subseteq [\p\WTD\text-\INC]^{\FPT}$.
\end{corollary}
\begin{proof}
	For the $\W t$-membership of $\p\WTD_{\varphi_t}$, notice that the translation between $\theta$ and the $\INC$-formula $\varphi_t$ in the proof of Lemma~\ref{INC:WT-level} preserves a one-to-one correspondence between the solutions $S$ for $\theta$ and satisfying teams $T$ for $\varphi_t$.
	In other words, $\theta$ has a solution of size $k$ if and only if $\varphi_t$ has a satisfying team of size $k$. 
	This yields $\W t$-membership since $\theta\in \PiForm{t}$ for each $t\geq 1$ (see Def.~\ref{def:wt}).
	The $\W t$-hardness and the containment $\W t \subseteq [\p\WTD\text-\INCPi t]^{\FPT}$ for all even $t\geq 1$ follows from Proposition~\ref{theorem-wsat}.
\end{proof}

We conclude this section by presenting the upper bounds for $\WTD_\varphi$ when $\varphi$ is an arbitrary $ \FO(\subseteq)$-formula.

\begin{theorem}\label{INC:WP}
 $[\p\WTD\text-\FO(\subseteq)]^{\FPT}\subseteq \W{\Ptime}$.
\end{theorem}
\begin{proof} 
We prove this via the machine characterization of the class $\W\Ptime$, analogous to the proof for $\FO$-formulas~\cite[Prop.~5.3]{DBLP:series/txtcs/FlumG06}.
Let $\varphi$ be a  $\FO(\subseteq)$-formula with $s$ free variables. An algorithm for the problem  $\p\WTD_\varphi$ proceeds as follows: Given a structure  $\calA$ and a $k$, nondeterministically guess $k$ times an assignment (i.e., an $s$-tuple of elements of $\calA$), then deterministically verify that the team $T$ has cardinality $k$ and $\calA,T\models \varphi$. 
Guessing $T$ requires $s \cdot k\cdot \log |A|$ nondeterministic bits, and the verification that   $\calA,T\models \varphi$  can be done in deterministic polynomial time in $|A|$~\cite{gallhella13}. 
Thus $\p\WTD_\varphi$ is in $\W\Ptime$ because the formula $\varphi$ is fixed and $s$ is a constant.
Moreover, the containment $[\p\WTD\text-\FO(\subseteq)]^{\FPT}\subseteq \W{\Ptime}$ holds since $\p\WTD_\varphi\in \W\Ptime$ for an arbitrary but fixed $\INC$-formula $\varphi$.
\end{proof}

\subsection{Dependence Logic}
	First observe that if $\varphi$ is a $\D$-sentence, then the problem $\p \WTD_\varphi$ is $\para\NP$-complete.
	This is due to the reason that the data complexity of fixed $\D$-sentences is already $\NP$-complete~\cite{DBLP:books/daglib/0030191}.

	\begin{theorem}\label{D:sentences}
		There is a $\D$-sentence $\varphi$, such that the problem $\p\WTD_\varphi$ is $\para\NP$-complete.
	\end{theorem}
 	\begin{proof}
 	Recall that a $\D$-sentence $\varphi$ has a satisfying team $T$ if and only if $\{\emptyset\}\models \varphi$.
 	For hardness, consider the data complexity of the model checking for $\D$-sentences. 
 	The problem asks whether an input structure $\calA$ satisfies a fixed $\D$-sentence $\varphi$.
 	Then $\varphi$ is true in $\calA$ if and only if $\calA,\{\emptyset\}\models \varphi$ if and only if there is a team $T$ such that $|T|=1$ and $\calA,T\models \varphi$.
 	\end{proof}
 	
%	\begin{corollary}\label{D:sentences}
%		Let $\Dsentence \subseteq \D$ denote the collection of sentences in $\D$, then $[\p\WTD\text-\Dsentence]^{\FPT}=\para\NP$.
%	\end{corollary}

% 	Regarding the complexity of $\WTD{}$ for $\D$-formulas, we first present a syntactic fragment $\Dfrag$ of $\D$.
% 	We will later prove that the problem $\WTD{}$ for $\Dfrag$ can encode problems complete for the $\W{}$-hierarchy.
%%	\begin{definition}
%%		Let $\alpha,\beta$ be $\FO$-formulas and $\depa{x}{y}$ be a dependence atom. Then the fragment $\Dfrag$ of $\D$ includes formulas $\varphi$ such that the quantifier free part of $\varphi$ has the form $(\alpha\lor (\beta\land\depa{x}{y}))$.
%%		In other words, $\varphi \dfn \forall x_1\exists x_2\dots \forall x_n \psi $, such that $\psi$ is $2$-coherent.
%		Let $\alpha$ be a $\FO$-formula and $\beta$ be a 2-coherent $\D$formula. Then the fragment $\Dfrag$ of $\D$ includes formulas $\varphi$ such that the quantifier free part of $\varphi$ has the form $(\alpha\lor \beta)$.
%		In other words, $\varphi \dfn \forall x_1\exists x_2\dots \forall x_n \psi $, such that $\psi$ is $2$-coherent.
%	\end{definition}
%
%Surprisingly, the data complexity of $\Dfrag$ is already $\NP$-complete, as proven by Durand~et~al.~\cite[Theorem~16]{DurandKRV22} 
%\begin{proposition}
%	There exists a formula $\varphi$ with an $\NP$-complete model-checking problem where $\varphi=\exists x \psi$ and $\psi$ is a conjunction of two dependence atoms.
%\end{proposition}

Now, we relate the $\W{}$-hierarchy to the weighted definability for dependence logic.
This also settles the complexity of $\p \WTD{}$ for $\D$-formulas.
%
%We start by presenting a reduction from the $\W1$-complete problem $\IndSet$. 
In the following, we prove that already one universal quantifier is enough in $\D$ to define $\W1$-complete problems.

\begin{theorem}\label{D:W1}
	There is a $\D$-formula $\varphi$ with only one universal quantifier such that the problem $\p\WTD_\varphi$ is $\W1$-complete.
\end{theorem}
\begin{proof}
	We present a reduction from the $\W 1$-complete problem $\p\IndSet$ to $\p\WTD_\varphi$ such that $\varphi$ is $\D$-formula with only one universal quantifier.
	An input to $\IndSet$ is a graph $\calG\dfn (V,E)$ and a number $k\in \mathbb N$.
	The question is whether there is a set $S$ of size $k$ in $\calG$ such that $(a,b)\not \in E$ for every $a,b\in S$. 
	We let $\tau \dfn \{N^1, P^1, I^2\}$ as our vocabulary where $N,P$ are unary relations %interpreting nodes and edges, 
	and $I$ is a binary relation symbol. % interpreting incidence relation between nodes and edges.
	Moreover the $\tau$-structure $\calA$ is such that: $\dom(\calA)\dfn V\cup E$, $N^\calA\dfn V, P^\calA\dfn E$ and $I^\calA$ simulates the edge relation $E^\calG$.
	That is, $I\dfn\{\,(a,b), (c,b)\mid a,c\in V,\text{ and } b\in P \text{ denotes the edge }(a,c)\in E\,\}$.
	Finally we define a $\D$-formula $\varphi$ over a single free variable $x$ as in the following.
	\[
	\varphi(x) \dfn \forall y \bigl(N(x) \land (\neg P(y)\lor \neg I(x,y) \lor \depas{y}{x}) \bigr)
	\]
%	Moreover, the formula $\varphi_{is}$ is fixed as it does not depend on the input instance $\calG$ of $\IndSet$.

	The correctness of our reduction is established via the following claim and also shows that the formula $\varphi$ is, in fact, equivalent to the familiar definition of independent sets via a $\PiForm{1}$-formula; hence, $\p\WTD_\varphi$ is $\W1$-complete.
	\begin{claim}
		There is a team $T$ over $x$ in $\calA$ such that $|T|=k$ and $\calA,T\models \varphi$ if and only if there is an independent set in $\calG$ of size $k$.\label{clm:indset}
	\end{claim}
	It remains to prove the claim.
	Suppose that $T=\{\,s_i\mid i\leq k\,\}$ is a team over $x$ for $\varphi$ such that $s_i(x)=a_i$ for $a_i\in A$.
	Moreover, let $T'=\{\,s_{i,j}\mid i\leq k,j\leq |\calA|\,\}$ denote the supplemented team, that is, $s_{i,j}(x)=a_i$ and $s_{i,j}(y)=a_j$ for every $a_j\in \calA$.
	We prove that $S=\{a_i \mid \exists s\in T, s(x)=a_i \}$ constitutes an independent set in $\calG$.
	Let $a_i,a_j\in S$, then there are $s_i,s_j \in T$ such that $s_i(x)=a_i$, $s_j(x)=a_j$.
	Suppose further that $(a_i,a_j)=e\in E^\calG$.
	Then, $T'\models P(e)\land E(a_i,e)$ and $T'\models E(a_j,e)$ but $T'\not\models \depas{y}{x}$ since there are $s_{i,j},s_{j,j}\in T'$ such that $s_{i,j}(xy)=a_ia_j$ and $s_{j,j}= a_ja_j$.
	In other words, $s_{i,j}(y)= s_{j,j}(y)$ but $s_{i,j}(x)\not= s_{j,j}(x)$.
	Consequently, $T'\not\models (\neg P(y)\lor \neg I(x,y)\lor \depas{y}{x})$ and $T\not\models \varphi$, which is a contradiction.
		
	Conversely, if there is an independent set $S$ of size $k$ in $\calG$ then we prove that $T\models \varphi(x)$ for $T =\{s_i \mid i\leq k, s_i(x)\in S\}$. 
	Clearly, the supplemented team $T'(x,y)$ has the following effect: for every $y$ that corresponds to an edge $e$ between elements $a_i,a_j\in A$, at most one of its endpoint $a_i$ or $a_j$ is in $T(x)$, which is the case if and only if $S$ is in independent set. \qedhere
\end{proof}

%Now suppose that $\tau$ contains only one binary relation but the number of unary relations is not fixed.
Once again, we prove the next lemma that generalizes the previous theorem to arbitrary levels of the $\W{}$-hierarchy. 
\begin{lemma}\label{D:WT-level}
Let $t\in \mathbb{N}$. Then there exists a fixed formula $\varphi_t\in \D$ with one free variable $z$ such that for all  $\psi \in \Gamma^{-}_{t,d}$ and $k\ge 1$: $\psi$ has a satisfying assignment of weight $k$ if and only if $A_{\psi},T\models \varphi_t$, for some team $T$ of cardinality $k$.
\end{lemma}

\begin{proof} 
Without loss of generality, we assume that $d=1$. 
Otherwise, the presented proof will easily generalize to larger values of $d$ by a disjunction/conjunction of arity $d$. 
By the results of \cite{KontinenV09}, it suffices to show that the required formula can be expressed by a first-order sentence $\theta(S)$ in which the relation symbol $S$ occurs only negatively. Then  the existence of  $\varphi_t(z)$ satisfying 
 \begin{equation}\label{eq:dep-wt}
 A_{\psi},T\models \varphi_t \Leftrightarrow  A_{\psi}\models \theta(S),
 \end{equation}
% $$A_{\psi},T\models \varphi_t \Leftrightarrow  A_{\psi}\models \theta(S),$$
 for all non-empty $T$ and $\rel(T)=S$ follows.
Now, it is easy to check that $\theta(S)$ can be expressed as follows:
 \begin{equation*}
 \theta(S)\dfn \forall x_1\big(\neg E(o,x_1) \vee \exists x_2(E(x_1,x_2)\wedge \cdots Qx_t(E(x_{t-1},x_t)\wedge I(x_t)\wedge \neg S(x_t)  )\cdots \big).
  \end{equation*}
The relation symbol $S$ appears only once in the formula and this appearance is negative.
%Now by Theorem~4.9 of \cite{KontinenV09}, there exists a formula $\varphi_t$ such \eqref{eq:dep-wt} holds for the sentence $\forall \vec{x}(S(\vec{x})\rightarrow \theta(S))$ for all $\calA$ and all $T$. 
%It is easy to see that $\theta(S)$ is equivalent with $\forall \vec{x}(S(\vec{x})\rightarrow \theta(S))$ modulo the cases when $S=\emptyset$.
%In fact, it is straightforward to show that $\varphi_t$ can be obtained from $\theta(S)$ simply by replacing $S(x_t)$ by a formula stating that $x_t$ is excluded from $z$ that can be expressed by an $\forall\exists$-dependence logic formula~\cite[Theorem~4.16]{galliani12} increasing the level by one.  
%The proof then is analogous to the proof of Theorem~\ref{D:W1}.
%
%
%\todo[inline]{Maybe, simply write $\D$-formula straight? The translated dependence formula will have arity $d$, to say that the  variables in the last clause/term are not all in the satisfying assignment. Or maybe add example with $t=3$ for DEP and $t=2$ for INC.}
\end{proof}
Notice further that the translation of the formula~$\theta$ to a $\D$-formula only introduces dependence atoms and, in particular, does not require any further quantification. 
Therefore, the following corollary (with proof analogous to Corollary~\ref{cor:inc-wt}) follows. % immediately from the proof in Lemma~\ref{D:WT-level}.
Recall that every dependence logic formula can be put into the $\forall\exists$-normal form. 
As a result, tracking the quantifier prefix in Lemma~\ref{D:WT-level} is not useful and we get the much stronger statement that the whole $\W{}$-hierarchy is already contained in $\DPi{2}$. 

\begin{corollary}\label{cor:D-wt}
Let $t\geq 1$ be odd. 
Then there is an $\D\text-\PiForm{2}$-formula $\varphi_t$ for which 
the problem $\p\WTD_{\varphi_t}$ is $\W t$-complete.
Moreover, $\bigcup_{t\ge1}\W t \subseteq [\p\WTD\text-\DPi2]^{\FPT}$.
\end{corollary}
Finally, $\D$ captures the class $\para\NP$ as established below.

\begin{theorem}\label{thm:WT-D}
 $[\p\WTD\text-\D]^{\FPT}=\para\NP$.
\end{theorem}

\begin{proof}Hardness follows from Theorem~\ref{D:sentences}. For membership,  we present the following non-deterministic algorithm that runs in polynomial time in the size of $\calA$.
Notice that since the formula is fixed, we have fixed many connectives including splits and existential quantifiers.
The idea of the algorithm is that it guesses a team $T$ of size $k$, as well as, a sequence $T_i$ for $i\in \mathbb N$ of teams which corresponds to the operations of duplication/supplementation and splits according to the formula $\varphi$.
In other words, let $\varphi=Q_1x_1Q_2x_2\dots Q_\ell x_\ell \psi$ where $Q\in\{\forall,\exists\}$ and $\psi$ is a quantifier free $\D$-formula.
Then the algorithm has the following steps.
\begin{itemize}
	\item Guess a team $T_0$ of size $k$ over $\Vars(\varphi)$.
	\item For each $i\leq \ell$, guess a team $T_i$ over $\Vars(\varphi)\cup\{x_1,\dots,x_i\}$ such that: if $Q_i = \forall$, then $T_i = P^x_{\calA}$ and if $Q_i = \exists$, then $T_i = P^x_{f}$ where $f\colon P\rightarrow \mathcal P(A)\setminus\emptyset $ and $P=T_{i-1}$.

	Notice that $|T_i| =|T_{i-1}|$ if $Q_i=\exists$ (because $\D$ is downwards closed) and $|T_i| =|T_{i-1}|\cdot |A|$ otherwise. 
	As a result, we have that $|T_i| \leq k\cdot |A|^i$.
	Once the team $T_\ell$ has been guessed, it remains to determine whether $T_\ell \models \psi$. 
	Since the data complexity of $\D$ is still $\NP$-complete for quantifier free formulas, this step is  non-trivial.
	Nevertheless, we can list recursively all the subformulas of $\psi$ in terms of its syntax tree.
	This helps in labelling a subteams of $T_\ell$ according to the connectives of $\psi$.
	\item Guess subteams for subformulas of $\psi$, such that: $T_\psi=T_\ell $. For each subformula $\alpha = \beta\land \gamma$ of $\psi$, $T_\beta=T_\gamma = T_\alpha$, and for each $\alpha = \beta\lor \gamma$, $T_\beta\cup T_\gamma = T_\alpha$.
	
	Clearly, the size of the subteam $T_\alpha$ for each $\alpha$ is atmost $k\cdot |A|^\ell$.
\end{itemize}

Notice that for atomic formulas the truth evaluation $T_\alpha\models \alpha$ can be determined in polynomial time.
Moreover, the intermediate steps including the verification of team duplication and supplementation can also be determined in polynomial time.
This results in $\para\NP$-membership of $\WTD_\varphi$ for a $\D$-formula $\varphi$.
%
%$\supseteq$: It is known that $\D$-sentences can express $\NP$-complete problems over finite structures. 
%In the following, we slightly modify $\D$-sentences into formulas with free variables.
%Let $\psi$ be a sentence of dependence logic expressing the problem 3-COL and define $\varphi\coloneqq (x=y\wedge \psi)$. Now for a given graph $G$ and a parameter $k$ (of a non-trivial size), $G,T\models \varphi$ for a team $T$ of size $k$ 
% iff $G$ is in 3-COL. Therefore,  the $k$-slice of the problem $\WTD_\varphi$ is $\NP$-complete for any fized $ k\ge 1$, and thus $\WTD_\varphi$ is $\para\NP$-hard.
%
\end{proof}

\subsection{Independence Logic}
In this section, we turn to independence logic. The following theorem is obtained from the results in the previous sections and the fact that any $\ESO$-sentence $\psi(S)$ (with an extra relation encoding the team) can be represented by an independence logic formula \cite{galliani12}. 

\begin{theorem}\label{thm:WT-I}
\begin{enumerate}
\item For all $t\in\mathbb N$ there is an $\IND$-formula $\varphi_t$ such that  $\p\WTD_{\varphi_t}$ is $\W t$-complete.
\item There is  an $\IND$-formula  $\varphi_w$ such that  $\p\WTD_{\varphi_w}$ is $\W\Ptime$-complete.
\item There is a $\IND$-sentence $\varphi$, such that the problem $\p\WTD_\varphi$ is $\para\NP$-complete.
\item $[\p\WTD\text-\IND]^{\FPT}=\para\NP$.
\end{enumerate}
\end{theorem}
\begin{proof}
%\todo[inline]{write further details here on 2.}

The first and third claim follow immediately from the fact that the logics $\D$ and $\INC$ are sublogics of $\IND$~\cite{galliani12} together with Theorem~\ref{thm:WT-D}. 

For the second claim, we use the fact that $\p\WSAT(\text{CIRC})$ is $\W\Ptime$-complete, where CIRC is the class of all propositional formulas encoded as Boolean circuits.
Note that the circuit value problem can be readily expressed by an $\ESO$-sentence $\psi(S)$, where $S$ represents an input for the circuit. 
More precisely, assume we a given a DAG $(A,E,D,K,I,o)$ encoding a Boolean circuit. Here $A$ is the set of nodes/gates,  $E$ is the edge relation, $I\subseteq A$ are the input gates of the circuit, $o$ is the unique output, $D\subseteq A$ is the set of OR-gates, and $K\subseteq A$  the set of AND-gates. A Boolean input for the circuit is represented by a subset $S\subseteq I$, i.e., a gate $g$ gets input $1$ if and only if  $g\in S$. Now in $\ESO$ we can existentially quantify a proof tree witnessing  the circuit accepting the input $S$. In other words, we quantify a subset $P\subseteq A$
 such that 
 \begin{itemize}
	\item $o\in A$,
	\item $P\cap I=S$,
	\item for all $g\in P\cap D$ there exists at least one $g'\in P$ such that $E(g',g)$,
	\item for all $g\in P\cap K$ and all $g'$, if $E(g',g)$ then $g'\in P$. 
\end{itemize}
It is straightforward to check that the above conditions can be  expressed  in first-order logic.

Finally, the hardness part of the third claim  follows again from the fact that $\D$ is a sublogic of $\IND$ and the containment proof is analogous to that of Theorem~\ref{thm:WT-D}.
\end{proof}

\section{Conclusion}

We have defined and studied the parameterized complexity of weighted team definability with respect to formulas of  several team-based logics. 
Our results show that for plain first-order formulas weighted team definability differs greatly from weighted Fagin definability; the former being computationally much simpler. 
For dependence, independence and inclusion logic formulas, the complexity of weighted team definability ranges between the classes $\W t$ and $\para\NP$. 
Now, these results provide a wide range of natural complete problems for the aforementioned complexity classes enriching the landscape in a nontrivial way. 
Interestingly, the sentences in the considered logics depict different complexities: namely, membership in $\FPT$ for $\INC$ and $\para\NP$-completeness for $\D$ and $\IND$.
The main open question is whether the converse directions of Corollary~\ref{cor:inc-wt} or Theorem~\ref{INC:WP} can be proven, i.e., if one of the inclusions $\bigcup_{t\in\mathbb N}\W t\subseteq [\p\WTD\text-\FO(\subseteq)]^{\FPT}\subseteq\W\Ptime$ is in fact an equality.
%%
%% Bibliography
%%

%% Please use bibtex, 

\bibliography{main.bib}

\appendix

\end{document}